\documentclass[11pt]{article}
\usepackage{graphicx}
\usepackage[misc,geometry]{ifsym}
\usepackage{verbatim}
\usepackage{amsfonts}
\usepackage{caption}
\usepackage{tabularx}
\usepackage{array}
\usepackage{booktabs}
\usepackage{multirow}
\usepackage{fancyvrb}
\usepackage{amsmath}

\usepackage{enumerate}

\usepackage{latexsym}
\usepackage{amsmath,amssymb}

\usepackage{color}

\newtheorem{theorem}{Theorem}
\newtheorem{lemma}{Lemma}
\newtheorem{corollary}{Corollary}

\def\ve{\varepsilon}

\def\E{\mathbf{E}}

\def\Pr{\mbox{{\bf Pr}}}
\def\Prob{\Pr}
\def\whp{{w.h.p.}}

\def\bigO{\mathcal O}

\newenvironment{proof}{\trivlist\item[]\emph{Proof}.}%
{\unskip\nobreak\hskip 1em plus 1fil\nobreak$\Box$
\parfillskip=0pt%
\endtrivlist}

{\unskip\nobreak\hskip 1em plus 1fil\nobreak$\Box$
\parfillskip=0pt%
\endtrivlist}
{\unskip\nobreak\hskip 2em plus 1fil\nobreak$\Box$
\parfillskip=0pt%
\endtrivlist}
{\unskip\nobreak\hskip 2em plus 1fil\nobreak$\Box$
\parfillskip=0pt%
\endtrivlist}

\newcommand{\ignore}[1]{}

\newcommand{\brac}[1]{\left( #1 \right)}

\newcounter{rot}

\begin{document}

\title{Fast plurality consensus in regular expanders
\thanks{This work was supported in part by EPSRC grant EP/M005038/1,
``Randomized algorithms for computer networks''. Nicol\'as~Rivera was supported by funding from Becas CHILE. Takeharu Shiraga was supported by JSPS KAKENHI Grant Number 15J03840.
Work carried out while Takeharu Shiraga
was visiting King's College London with the support of the  ELC project (Grant-in-Aid for Scientific Research on Innovative Areas MEXT Japan).
}}

\author{
Colin Cooper\thanks{Department of Informatics, King's College London, UK.
{\tt colin.cooper@kcl.ac.uk}}
\and Tomasz Radzik\thanks{Department of Informatics, King's College London, UK.
{\tt tomasz.radzik@kcl.ac.uk}}
\and Nicol\'as Rivera\thanks{Department of Informatics, King's College London, UK.
{\tt nicolas.rivera@kcl.ac.uk}}
\and Takeharu Shiraga\thanks{
Theoretical Computer Science Group, Department of Informatics, Kyushu University, Fukuoka, Japan.
{\tt shiraga@tcslab.csce.kyushu-u.ac.jp}}
}

\maketitle

\begin{abstract}
The problem of reaching  consensus in a graph by means of local interactions
is an abstraction of such behavior in human society as well as some distributed processes in computer networks. In a {\em voting process} on a graph vertices revise their opinions in a distributed way based on the opinions of  nearby vertices. The classic example is synchronous pull voting where
at each step, each vertex  adopts the opinion of a random neighbour.
This type of pull voting suffers from two main drawbacks.
Even if there are only two opposing opinions, the time taken for a single opinion to emerge can be slow, and the final opinion is not necessarily the initial majority.
Things can often be improved by
using  a variant of synchronous pull voting
in which each vertex considers the opinions of two neighbours.
For many classes of $n$-vertex regular expanders,
consensus is now reached in $O(\log n)$
expected steps~\cite{CER},
as opposed to $\Theta(n)$ expected steps~\cite{CEHR-SIAM2013} when only one neighbour is contacted.
Moreover, this protocol allows the
initial majority opinion to win with high probability.

In the case where there are initially  three or more opinions, not so much is known about the performance of  voting using two or more samples.
A problem arises when there is no clear majority. Thus one class of opinions may be largest, but its total size is less than that of two other opinions put together. When there are three or more opinions, the term {\em plurality} is often used to distinguish this case from that of an overall majority.

In the case where the underlying network is the complete graph $K_n$,
Becchetti et.\ al~\cite{becchetti2014plurality, becchetti2014plurality_arXiv}
analysed the general case of $k \ge 3$ opinions using a {\em three-sample voting process} and proved the following result.
Let $A_1$ be the
initial size of the largest opinion.
Then if the difference between the
initial sizes  of the largest and second largest opinions
is at least $C n  \sqrt{(\log n)/A_1}$, for some  suitable constant $C$,
the largest opinion wins in $O((n \log n)/A_1)$ steps
with high probability.

In this paper we show that similar performance can be achieved
on  {\em $d$-regular expanders}
using {\em two-sample voting}.
Namely, if
 the difference between the
initial sizes  of the largest and second largest opinions
is at least
$ C n \max\{\sqrt{(\log n)/A_1}, \lambda\}$,
for some  suitable constant $C$, then
the largest opinion wins in $O((n \log n)/A_1)$ steps
with high probability. Here $\lambda$ is
the absolute second eigenvalue of transition matrix $P=Adj(G)/d$
of a simple random walk on the graph $G$.
For almost all $d$-regular graphs,
we have $\lambda=c/\sqrt{d}$ for some  constant $c>0$ \cite{F}.
Thus as $d$ increases we can separate an opinion whose plurality is $o(n)$,
whereas a plurality of $\Theta(n)$  appears to be needed for $d$ constant.
Finally  for $d$ constant we show how this $\Theta(n)$ barrier can be reduced
by sampling using short random walks.
\end{abstract}

%
%

\section{Introduction}

The problem of reaching  consensus in a graph by means of local interactions
is an abstraction of such behavior in human society as well as some processes in computer networks. In a {\em voting process} on a graph,  vertices revise their opinions in a systematic and distributed way based on  opinions of  other vertices, typically using a sample of their local neighbours.
The aim is that eventually a single opinion will emerge, and that this opinion will reflect the relative importance of the original mix of opinions in some way.

Voting processes  
are a natural approach to achieving consensus,
and as a consequence they have been widely studied.
Distributed voting finds  application in various
fields of computing including consensus and leader election in large networks
\cite{BMPS04, HassinPeleg-InfComp2001},
serialisation of read/write in replicated data-bases \cite{Gifford79},
and analysis of social behavior \cite{DP94}.
In general,  a voting process should  be conceptually simple, fast,
fault-tolerant and straightforward to implement \cite{HassinPeleg-InfComp2001, Joh89}.

In outline, a voting process can be described as follows. Each vertex of a connected graph
has one of several possible opinions. In each time-step,  each vertex
 queries the opinion one or more of its neighbours  using the same protocol, and decides whether to modify or to keep its current opinion.
When
all vertices have a common (and thus final) opinion, we say a consensus has been reached.
For a given voting process, the main questions of interest are the probability that a particular opinion wins
and the expected time to reach consensus.
The most well known model is synchronous pull voting.
In this model,  at each step each vertex changes its opinion to that of a random neighbour.

In the classical {\em voter model} each vertex initially has a distinct opinion,
but in general we can assume the vertices are restricted to hold one of $k$ different opinions.
The simplest case, {\em two party voting}, is when there are initially two opinions ($k=2$).
If there are at least three opinions ($k \ge 3$) the problem is often referred to as {\em plurality consensus}.
Not so much is known about improving the performance of  voting by using two or more samples in the case where there are initially  three or more opinions.

If some opinion has an absolute majority, we can group the other opinions together into a single minority class, and use the above two-sample protocol.
A problem arises when there is no clear majority. Thus one class of opinions may be largest, but its total size is less than that of two other opinions put together. When there are three or more opinions the term {\em plurality} is often used to distinguish this case from the overall majority one.

For the problem to be one of plurality consensus, we assume that the initial configuration is such that one opinion is dominant, but there is no overall majority.
We might expect that
the  dominant opinion eventually becomes
the final opinion of all vertices.
This, however, strongly depends on the voting process.
If simple pull voting is used, then given the graph is
connected (and aperiodic) the probability that a particular opinion wins
is proportional to the initial degree of the opinion in the graph (see \cite{HassinPeleg-InfComp2001}).
More precisely, if $A$ is the set of vertices initially holding a given opinion, then
the probability $A$ wins in the voting process is
\begin{equation}\label{Awins}
\Pr(A\text{ wins})= \sum_{v \in A} \frac{d(v)}{2m}= \frac{d(A)}{2m},
\end{equation}
where $d(v)$ is the degree of vertex $v$ and $m$ is the number of edges in the graph.
Surprisingly, the probability here depends only on the voting process and not on the initial arrangement of opinions on the graph (any set of vertices of the same total degree would do).

We assume henceforth that the graphs we consider are connected and that the graph is not bipartite, so that a consensus is possible.
For an $n$-vertex graph, let  $\E T= \E T(n)$ be the expected value of the time to consensus $T$.
Much of the early work was on
analysing t $\E T$ for classical pull voting in an asynchronous model in a continuous time setting. Here the vertices have independent exponentially distributed waiting times
(Poisson clocks); see e.g.  Cox~\cite{Cox-1989} and Aldous~\cite{Aldous-MeetingTimes}.
In the synchronous model the expected time to consensus can be bounded by
$\E T = O(H_{\max} \log n)$, where $H_{\max}=O(n^3)$ is the maximum hitting time of
any vertex by a random walk; see Aldous and Fill~\cite{AlFi}.
For regular expanders these results can be improved to $\E T = \Theta(n)$, see \cite{CEHR-SIAM2013}.

Because the classical pull voting tends to be slow ($\E T =\Theta(n)$ for regular expanders)
and may be viewed as undemocratic,
there has been considerable interest in modifying this simple voting process to avoid these two problems.
Instead of taking the opinion of only one neighbour, the next simplest approach to sample the opinions of a larger number of neighbours
(say two or three),
compare them in some way, and hope that the so-called \lq power of two choices\rq\ improves
the performance of voting.
The consequences of this approach are as follows.
Firstly, the number of
neighbours queried affects the consensus time and the voting outcome.
Secondly,  the relative size of the opinions
 affects
the ability of the process to ensure that the largest initial opinion wins.
Not surprisingly, analysing this relation becomes harder when we move from two party voting to
plurality consensus ($k\ge 3$).
The additional challenge is that
the well established techniques used in analysis of the classical pull voting
(for example, the correspondence with multiple coalescing
random walks~\cite{Aldous-MeetingTimes, CEHR-SIAM2013})
do not have ready extensions or generalisations to multi-sample voting.

In this setting we study the following protocols for two-sample and three-sample voting.
In the two-sample voting model, at each step, each vertex $v$ chooses two random neighbours
with replacement, and if the selected vertices have the same the opinion, then $v$ adopts it;
otherwise $v$ keeps its current opinion.
In the three-sample voting model, each vertex $v$
chooses three random neighbours with replacement, and $v$ adopts the majority opinion among them.
If there is no majority, $v$ picks the opinion of the first sampled neighbour.
Other rules are equally possible here, e.g. $v$ keeps its opinion.
The rule we choose is the one used by  Becchetti et. al.  \cite{becchetti2014plurality,becchetti2014plurality_arXiv}, and we adopt it for consistency.

Two-sample voting was studied in \cite{CER} for the case where there are
initially two opinions ($k = 2$). They proved that in $d$-regular expanders
the initial majority wins with high probability (\whp)\footnote{``With high probability'' (\whp) means in this paper probability at least $1 - n^{-\alpha}$, for
a constant $\alpha > 0$.}
provided the initial difference between the sizes of the two opinions is sufficiently large,
and that voting is completed in $\bigO(\log n)$ steps. This is tight since the diameter of a $d$-regular
graph is $\Omega(\log n)$ for constant $d$. In \cite{cooper2015fast} the authors extend the above
result to general expander graph, extending the analysis to non-regular graph.

As hinted at above the analysis for plurality consensus ($k \ge 3$)
tends to be trickier than for two party voting.
This is especially true as $k$ increases,
or if two minorities together are much larger then the majority opinion.
Plurality consensus using the three-sample voting protocol given above
was studied by
Becchetti et. al.  \cite{becchetti2014plurality,becchetti2014plurality_arXiv}.
They proved that for the complete graph $K_n$, if the difference between the
initial sizes $A_1$ and $A_2$ of the largest and second largest opinions
is at least $A_1-A_2= 24 n  \sqrt{2(\log n)/A_1}$,
then the largest opinion wins in $O((n \log n)/A_1)$ steps \whp\
They also showed that this result is tight for some ranges of the parameters.

\subsection{Our contributions}

In this paper we extend the results of~\cite{becchetti2014plurality,becchetti2014plurality_arXiv} from the  complete graph to $d$-regular expanders preserving
 the same asymptotic convergence time. To do this, we generalize
the results of~\cite{cooper2015fast} from two-party voting to $k$-party voting.
We also give a natural coupling of the three-sample process
of~\cite{becchetti2014plurality,becchetti2014plurality_arXiv} with the two-sample process of \cite{cooper2015fast},
which allows us to apply our analysis of the two-sample process directly to the three-sample process.

We proceed to state our main result.
Let $G$ be a connected regular $n$-vertex graph and let $\lambda$ be
the second largest absolute eigenvalue of the transition matrix $P=P(G)$ of a random walk on $G$.
Let $A_1$ be the set of vertices with the largest initial  opinion
and $A_2$ the set with the second largest opinion.
If no confusion arises, we also use $A$  to stand for the size of  set $A$.

\begin{theorem}\label{thm:mainResult}
Let $G$ be a regular $n$-vertex graph and
let the initial sizes of the opinions be $A_1,A_2,\ldots,A_k$ in non-increasing order.
Assume that
$A_1-A_2 \geq C n \max\{ \sqrt{(\log n )/A_1}, \lambda\}$,
where $\lambda$ is the absolute second eigenvalue  of $P(G)$ and
$C > 0$ is a suitably large constant.\\
With probability at least
$1- 1/n$, after at most $O((n/A_1) \log(A_1/(A_1 - A_2)) + \log n)$ rounds,
the two-sample voting completes and the final opinion is the largest initial opinion.
\end{theorem}
We note the following w.h.p. property of the second eigenvalue $\lambda$ for  random $d$-regular graphs for $d=o(n^{1/2})$.
For $d$ constant
it is a result of Friedman~\cite{F} that
$\lambda \le \gamma/\sqrt{d}$, where $\gamma=2+\epsilon$ for some small $\epsilon >0$.
For $d$ growing with $n$,
the following estimate of $\lambda$ is given in~\cite{BFSU}.
Provided $d =o(n^{1/2})$ there exists constant $\gamma>0$
such that \whp\ $\lambda \le \gamma/\sqrt{d}$. In either case the size separation condition in Theorem \ref{thm:mainResult} is $A_1-A_2 \geq C'n/\sqrt{d}$.

Theorem \ref{thm:mainResult} can be applied to a number of specific scenarios.
Consider, for example, the case where all $k$ opinions are fairly evenly represented,
but with one opinion slightly larger than
the average $n/k$.
More specifically, assume that $A_1 \ge (n/k)(1+\varepsilon)$, for some $0 < \varepsilon \le 1$,
and that $A_2 \le A_1/(1+\varepsilon)$.
Theorem~\ref{thm:mainResult} implies the following corollary for this case.

\begin{sloppy}

\begin{corollary}\label{nxjsbx781}
For $k \le ( (1/C)^2 n /\log n)^{1/3}$ opinions,
if
$A_1 \ge (n/k)(1 + \varepsilon)$,
$A_2 \le A_1/(1 + \varepsilon)$,
and $\lambda \le \varepsilon/(Ck)$,
where $C > 0$ is the constant from Theorem~\ref{thm:mainResult}
and $\varepsilon^{2/3} = k / ( (1/C^2) n /log n)^{1/3} \le 1$.\\
With probability at least
$1- 1/n$,
after at most
$O(k\log n)$
rounds
the two-sample voting completes and the final opinion is the largest initial  opinion.
\end{corollary}
\end{sloppy}
In Section~\ref{3to2} we show
that the statements of Theorem~\ref{thm:mainResult} and Corollary~\ref{nxjsbx781}
also hold for the three-sample voting protocol used  by Becchetti et.\ al.~\cite{becchetti2014plurality,becchetti2014plurality_arXiv}.
We note that the bound on the running time in Theorem~\ref{thm:mainResult}
is $O(\log n)$, if $A_1$ is $\Omega(n/\log n)$, provided that
$A_1-A_2$ is also $\Omega(n/\log n)$ and
$\lambda$ is appropriately small.
This improves on the results of~\cite{becchetti2014plurality,becchetti2014plurality_arXiv}
which require $A_1 = \Theta(n)$ for a running time  of $O(\log n)$.

In the $\ell$-extended two-sample voting model, (as introduced in
 \cite{cooper2015fast}) each vertex makes
two independent random walks of length $\ell$ and carries out two-sample
voting using
the opinions on the terminal vertices of these walks.
By sampling  using  random walks of length $\ell$, we replace the transition matrix $P$
used in the proof of Theorem~\ref{thm:mainResult}
by its $\ell$-th power $P^\ell$. If the graph is regular, then the only effect on the proofs is to replace all eigenvalues by their $\ell$-th power.
This reduces the  absolute second eigenvalue
from $\lambda$ to $\lambda^\ell$.
By increasing $\ell$ we can include in our analysis those graphs which do not satisfy the conditions of Theorem~\ref{thm:mainResult} on the relation between $A_1 - A_2$ and $\lambda$.

\begin{theorem}\label{thm:mainResult3}
Let  $\ell$ be a positive integer,
let $G$ be a regular $n$-vertex graph and
let the initial sizes of the opinions be $A_1,A_2,\ldots,A_k$ in non-increasing order.
Assume that $A_1-A_2 \geq C n \max\{ \sqrt{(\log n )/A_1}, \lambda^\ell\}$,
where
$C > 0$ is the constant from Theorem~\ref{thm:mainResult}.
Then
$\ell$-extended two-sample voting completes
after at most $O((n/A_1) \log(A_1/(A_1 - A_2)) + \log n)$ rounds,
with probability at least
$1- 1/n$,
and the final opinion is the largest initial opinion.
\end{theorem}
Once again the same results apply to $\ell$-extended three-sample voting.

\section{Preliminary Markov chain results}

In this section we set up some Markov chain foundations and preliminary results which we need for our proof of
Theorem~\ref{thm:mainResult}.
Consider a connected and non-bipartite graph $G = (V,E)$ with $n$ vertices and $m$ edges.
Let $P$ be the transition matrix of a simple random walk  on $G$.
A  random walk on a connected and non-bipartite graph defines
a reversible Markov chain with stationary distribution $\pi(x) = d(x)/(2m)$, where $d(x)$ denotes the degree of vertex $x$.
The reversibility of $P$ means that $\pi(x)P(x,y)=\pi(y)P(y,x)$, for all vertices $x,y$.

Let $1=\lambda_1 >\lambda_2 \geq \ldots \geq \lambda_n > -1$ be the eigenvalues of $P$ and define $\lambda=\lambda(P)$ by $\lambda = \max\{|\lambda_2|, |\lambda_n|\}$.
We also consider the matrix $P^2= P \times P$ (standard matrix product),
which is the transition matrix of the two-step random walk,
is also reversible and has the same stationary distribution and eigenvectors as $P$.
Moreover, the eigenvalues of $P^2$ are the squares of the eigenvalues of $P$.
In particular, $\lambda(P^2)=(\lambda(P))^2$.
Given $A, B\subseteq V$ and $x \in V$,
we define $P(x,A) = \sum_{y \in A} P(x,y)$ and the {\em flow function\/} $Q(A,B)$ from $A$ to $B$ as
\begin{eqnarray}\label{eqn:defQ(A,B)}
Q(A,B) = \sum_{x \in A} \pi(x)P(x,B).
\end{eqnarray}
The value of $Q(A,B)$ is the probability that
one step of the random walk taken from the stationary distribution
is a transition from a vertex in $A$ to a vertex in $B$.
Due to reversibility of $P$, $Q(A,B) = Q(B,A)$.
We will use  the following inequalities, sometimes known as  the
{\em Expander Mixing Lemma for Inhomogeneous Graphs} (see e.g.~\cite{cooper2015fast,Trev}).
Let $A, B \subseteq V$, and $A^c=V \setminus A$, then
\begin{eqnarray}\label{lemma:mixingQAAc}
|Q(A,A^c)-\pi(A)\pi(A^c)| & \leq & \lambda \pi(A)\pi(A^c), \\
\label{lemma:mixingQAB}
|Q(A,B)-\pi(A)\pi(B)| & \leq & \lambda \sqrt{\pi(A)\pi(B)\pi(A^c)\pi(B^c)}.
\end{eqnarray}
We also need  lower bounds for $Q^2$.
\ignore{
\begin{lemma}\label{lemma:Q(A,A)^2}
Let $A \subseteq V$, then~
$Q(A,A)^2 \; \geq \; \pi(A)^4 -2\lambda \pi(A)^2\pi(A^c)^2.$
\end{lemma}
\begin{proof}
Assume $Q(A,A) \leq \pi(A)^2$, otherwise the result is immediate. Observe that $Q(A,A) = \pi(A)-Q(A,A^c)$ which, together with our assumption, implies that $Q(A,A^c) \leq \pi(A)\pi(A^c)$.
Therefore, using~(\ref{lemma:mixingQAAc}), we derive
\begin{eqnarray}
Q(A,A)^2 &=&  \pi(A)^2-2\pi(A)Q(A,A^c)+Q(A,A^c)^2 \nonumber\\
&\geq&\pi(A)^2-2\pi(A)^2\pi(A^c) + (1-\lambda)^2\pi(A)^2\pi(A^c)^2\nonumber\\
&\geq&\pi(A)^2-2\pi(A)^2\pi(A^c) + (1-2\lambda)\pi(A)^2\pi(A^c)^2\nonumber\\
&=& (\pi(A)-\pi(A)\pi(A^c))^2 -2\lambda\pi(A)^2\pi(A^c)^2 = \pi(A)^4 -2\lambda \pi(A)^2\pi(A^c)^2.
\end{eqnarray}
\end{proof}
}

\begin{lemma}\label{lemma:Q(A,B)^2}
For any $A,B \subseteq V$, we have
\begin{equation}\label{jvwe891x}
Q(A,B)^2 \geq (\pi(A)\pi(B))^2 -  2 \lambda (\pi(A)\pi(B))^{3/2} (\pi(A^c)\pi(B^c))^{1/2}.
\end{equation}
\end{lemma}

\begin{proof}
\begin{eqnarray}
Q(A,B)^2 &=& ((Q(A,B)-\pi(A)\pi(B))+\pi(A)\pi(B))^2 \nonumber\\
&=& (Q(A,B)-\pi(A)\pi(B))^2+(\pi(A)\pi(B))^2+2\pi(A)\pi(B)(Q(A,B)-\pi(A)\pi(B))\nonumber\\
&\ge& (\pi(A)\pi(B))^2-2\lambda \pi(A)\pi(B) \sqrt{\pi(A)\pi(B)\pi(A^c)\pi(B^c)}.
\end{eqnarray}
The last line follows from \eqref{lemma:mixingQAB}.
\end{proof}
Given $A,B \subseteq V$, define the quantity $R(A,B) = \sum_{x \in A} \pi(x)(P(x,B))^2$.
This quantity is the expected change, in the stationary measure $\pi$,
from $A$ to $B$ in one round of two-sample voting.
\begin{lemma}\label{lemma:RVA = Q2AA}
For any $A \subseteq V$, we have~ $R(V,A) \: = \: Q_2(A,A),$~
where $Q_2$ is the flow function for the two-step transition matrix  $P^2$.
\end{lemma}

\begin{proof}
From definition of $R(V,A)$,
reversibility of $P$ and $P^2(x,y)= \sum_{z \in V}P(x,z)P(z,y)$:
\begin{eqnarray}
R(V,A) &=& \sum_{x \in V} \pi(x)P(x,A)^2 = \sum_{x \in V} \pi(x)P(x,A)\sum_{y \in A}P(x,y) = \sum_{y \in A}\sum_{x \in V} \pi(x)P(x,A)P(x,y)\nonumber \\
&=&\sum_{y \in A}\sum_{x \in V} \pi(y)P(y,x) P(x,A) = \sum_{y \in A} \pi(y)\sum_{x \in V}P(y,x) P(x,A)\nonumber\\
&=& \sum_{y \in A} \pi(y) P^2(y,A) = Q_2(A,A).
\end{eqnarray}
\end{proof}

If $G$ is a complete graph (with node loops), then
$R(V,A) = \pi(A)^2 = (|A|/n)^2$ and
$R(A,B) = \pi(A)\pi(B)^2 = |A|\cdot |B|^2/n^3$.
The next two lemmas give bounds on deviations from these values in regular graphs.

\begin{lemma}\label{lemma:R(V,A)-piA2}
For $A \subseteq V$, we have
\begin{equation}\label{mnd78a}
|R(V,A) -\pi(A)^2| \; = \; |Q_2(A,A^c)-\pi(A)\pi(A^c)| \; \leq \; \lambda^2\pi(A)\pi(A^c).
\end{equation}
\end{lemma}

\begin{proof}
By Lemma~\ref{lemma:RVA = Q2AA}, $R(V,A) = Q_2(A,A)$, and
standard manipulations give $Q_2(A,A) = Q_2(A,V)-Q_2(A,A^c) = \pi(A)-Q_2(A,A^c)$, so
$$R(V,A) - \pi(A)^2 = \pi(A)-Q_2(A,A^c)-\pi(A)^2 = \pi(A)\pi(A^c)-Q_2(A,A^c).$$
Taking the absolute value of both sides gives the first equality in~\eqref{mnd78a}.
To obtain the inequality, apply~\eqref{lemma:mixingQAAc} to $P^2$, $Q_2$ and
$\lambda^2$ as the second largest absolute eigenvalue of $P^2$.
\end{proof}

\ignore{
\begin{lemma}\label{lemma:R(A,A)>}
Let $A\subseteq V$, then
$$R(A,A) \geq \frac{Q(A,A)^2}{\pi(A)} \geq \pi(A)^3-2\lambda\pi(A)\pi(A^c)^2.$$
\end{lemma}
\begin{proof}
The second inequality is from Lemma~\ref{lemma:Q(A,A)^2}.
From convexity of the function $z \mapsto z^2$,
\begin{equation}\label{jnvjkd0sq}
R(A,A) = \sum_{x \in A} \pi(x)(P(x,A))^2 \ge \pi(A) \left( \sum_{x \in A} \frac{\pi(x)}{\pi(A)} P(x,A)\right)^2 = \frac{1}{\pi(A)} (Q(A,A))^2.
\end{equation}
\end{proof}
}

\begin{lemma}\label{lemma:R(A,B)>}
Let $A,B \subseteq V$, then
\[ R(A,B) \; \geq \; \frac{Q(A,B)^2}{\pi(A)} \geq \pi(A)\pi(B)^2
                        - 2\lambda \pi(A)^{1/2}\pi(B)^{3/2}\pi(A^c)^{1/2}\pi(B^c)^{1/2}.
\]
\end{lemma}
\begin{proof}
The second inequality is from Lemma~\ref{lemma:Q(A,B)^2}.
From convexity of the function $z \mapsto z^2$,
\begin{equation}\label{jnvjkd0sq}
R(A,B) =\pi(A)\sum_{x \in A} \frac{\pi(x)}{\pi(A)}(P(x,B))^2
\ge \pi(A) \left( \sum_{x \in A} \frac{\pi(x)}{\pi(A)} P(x,B)\right)^2 = \frac{1}{\pi(A)} Q(A,B)^2.
\end{equation}
\end{proof}

Suppose the family of sets $\mathcal C = (A_1, \ldots, A_k)$ is a partitioning of $V$. Define the quantity $S_{\mathcal C}(A) = \sum_{i = 1}^k R(A,A_i)$.
For a complete graph, $S_{\mathcal C}(A) = \sum_{i = 1}^k \pi(A_i)^2$ and the following
lemma bounds the deviation from this value in regular graphs.

\begin{lemma}\label{lemma:S(V)}
Consider a partitioning $\mathcal C = (A_1, \ldots, A_k)$ of $V$. Then
$$\left|S_{\mathcal C}(V)- \sum_{i=1}^k\pi(A_i)^2 \right|
   \leq \lambda^2\left(1-\sum_{i=1}^k \pi(A_i)^2\right).$$
\end{lemma}
\begin{proof}
Using Lemma~\ref{lemma:R(V,A)-piA2}, we get
\begin{eqnarray}
\left|S_{\mathcal C}(V)- \sum_{i=1}^k\pi(A_i)^2 \right|
    &=& \left|\sum_{i=1}^k R(V,A_i)-\pi(A_i)^2\right|
    \; \leq \; \sum_{i=1}^k \left|R(V,A_i)-\pi(A_i)^2\right| \nonumber\\
    &\leq& \sum_{i=1}^k \lambda^2\pi(A_i)\pi(A_i^c)
    \; = \; \lambda^2\left(1-\sum_{i=1}^k \pi(A_i)^2\right).\nonumber
\end{eqnarray}
\end{proof}

\begin{lemma}\label{lemma:BoundsS_c(A)}
Let $\mathcal C = (A_1, \ldots, A_k)$ be a partitioning of $V$.
For any $A \subseteq V$,
\begin{eqnarray}
S_{\mathcal C}(A) & \geq & \pi(A)\sum_{i=1}^k \pi(A_i)^2
           - 2\lambda \pi(A)^{1/2}\pi(A^c)^{1/2}\sum_{i=1}^k \pi(A_i)^{3/2},  \label{bgvhdv788cb1} \\
S_{\mathcal C}(A) & \leq & \pi(A)\sum_{i=1}^k \pi(A_i)^2
      +2\lambda \pi(A)^{1/2}\pi(A^c)^{1/2}\sum_{i=1}^k \pi(A_i)^{3/2}
      +\lambda^2. \nonumber
\end{eqnarray}
\end{lemma}

\begin{proof}
Lemma~\ref{lemma:R(A,B)>} gives the first part:
\[
S_{\mathcal C}(A) \; = \; \sum_{i=1}^k R(A,A_i) \geq \pi(A)\sum_{i=1}^k \pi(A_i)^2
     - 2\lambda \pi(A)^{1/2}\pi(A^c)^{1/2}\sum_{i=1}^k \pi(A_i)^{3/2}.
\]
For the second part,
observe that $S_{\mathcal C}(A) + S_{\mathcal C}(A^c) = S_{\mathcal C}(V)$ and use Lemma~\ref{lemma:S(V)}
and~\eqref{bgvhdv788cb1}:
\begin{eqnarray}
\lefteqn{S_{\mathcal C}(A) = S_{\mathcal C}(V)-S_{\mathcal C}(A^c)} \nonumber\\
&\leq& \sum_{i=1}^k \pi(A_i)^2
         + \lambda^2\left(1-\sum_{i=1}^k\pi(A_i)^2\right)
         - \pi(A^c)\sum_{i=1}^k\pi(A_i)^2
        +2\lambda \pi(A^c)^{1/2}\pi(A)^{1/2}\sum_{i=1}^k \pi(A_i)^{3/2}\nonumber\\
&=& \pi(A)\sum_{i=1}^k \pi(A_i)^2
       +2\lambda \pi(A)^{1/2}\pi(A^c)^{1/2}\sum_{i=1}^k \pi(A_i)^{3/2}
       +\lambda^2\left(1-\sum_{i=1}^k\pi(A_i)^2\right).\nonumber
\end{eqnarray}
\end{proof}

\section{Proof of Theorem~\ref{thm:mainResult}}

From now on we assume the graph is $d$-regular, so $\pi(x) = 1/n$, and for $A \subseteq V$, $ \pi(A)=|A|/n$.
Furthermore, $n R(A,B)= \sum_{x \in A} (d_B(x)/d)^2$ is the expected number of
vertices in $A$ which pick two opinions in $B$.
When  clear from the context, we use $A$ instead of $|A|$ for the size of $A$.

Let $A_j$
be the set of vertices with opinion $j$. At any step, the opinions are ordered according to their sizes:
$A_1 \geq A_2 \geq \ldots \geq A_k$.
Thus ${\mathcal C} = \{A_1,\ldots, A_k\}$ is a partition of $V$.
Let $A_j'$ be the set of vertices with opinion $j$ after one round.
We have the following equality, where the second term in~\eqref{bckdcb90q} is the expected number of
vertices changing their opinion to $A_j$ and the third term is the expected number of vertices
changing their opinion from $A_j$.
\begin{eqnarray}
\E(\pi(A_j')|\mathcal C) & = & \pi(A_j) + R(V\setminus A_j ,A_j) - \sum_{i\neq j} R(A_j,A_i) \label{bckdcb90q} \\
& = & \pi(A_j) + R(V,A_j) - R(A_j ,A_j)- \sum_{i\neq j} R(A_j,A_i) \nonumber \\
& = & \pi(A_j) +R(V,A_j)-S_{\mathcal C}(A_j).\label{eqn:ExpectedDiff}
\end{eqnarray}

The next lemma shows that, given a sufficient advantage of opinion $1$,
after one round of voting opinion $1$ remains the largest opinion.
More precisely, the lemma gives lower bounds
on the increase of the size of opinion $1$ and on the increase of the advantage of this opinion
over the other opinions.

\begin{sloppy}

\begin{lemma}\label{lem:main-oneStep}
Assume $A_1 \le 2n/3$, $A_1-A_2 \geq C n \sqrt{(\log n )/A_1}$
(requiring $A_1 \ge C^{2/3} n^{2/3} \log^{1/3} n$),
where $C = 240 \sqrt{2}$, and $\lambda \leq {(A_1-A_2)}/{(32n)}$.
Then with probability at least
$1- 1/n^2$,
\begin{eqnarray}
A_1' & \ge &  A_1\brac{1+\frac{A_1-A_2}{5n}}. \label{xi}\\
\min_{2 \le j \le k} \left\{A_1'-A_j'\right\} &\geq& (A_1 - A_2)\brac{1+\frac{A_1}{10n}},\label{eqn:Delta}
\end{eqnarray}
\end{lemma}

\end{sloppy}

%
%

\begin{proof}
Several times in this proof we use that $\pi(A_1)\leq 2/3$,
which implies that $\pi(A_1^c) \geq 1/3$.
Our proof uses concepts from Bechetti et.~al.~\cite{becchetti2014plurality,becchetti2014plurality_arXiv} and makes
extensive use of  the following Chernoff bounds.
If $X$ is the sum of independent Bernoulli random variables, then for $\varepsilon \in (0,1)$ and $\delta \ge 1$,
\begin{eqnarray}
\Prob(X \geq (1+\varepsilon)\E(X)), \; \Prob(X \leq (1-\varepsilon)\E(X))
  & \leq & \exp(-\varepsilon^2\E(X)/3), \label{eqn:k33dz} \\
\Prob(X \geq (1+\delta)\E(X))
  & \leq & \exp(-\delta\E(X)/3). \label{eqn:k99dz}
\end{eqnarray}
From Equation~\eqref{eqn:ExpectedDiff} and Lemmas~\ref{lemma:R(V,A)-piA2}
and~\ref{lemma:BoundsS_c(A)},
we have
the following lower and upper bounds on $\E(\pi(A_j')|\mathcal C)$
for any $j\in [k]$.
\begin{eqnarray}
\lefteqn{\E(\pi(A_j')|\mathcal C) \; = \; \pi(A_j) + R(V,A_j)-S_{\mathcal C}(A_j)} \nonumber\\
&\geq& \pi(A_j) +\pi(A_j)^2-\lambda^2\pi(A_j)\pi(A_j^c) \nonumber\\
&&- \pi(A_j)\sum_{i=1}^k \pi(A_i)^2-2\lambda\pi(A_j)^{1/2}\sum_{i=1}^k\pi(A_i)^{3/2}
      - \lambda^2\nonumber\\
&\ge &\pi(A_j)\left(1+\pi(A_j)-\sum_{i=1}^k\pi(A_i)^2 \right)
   - 2\lambda\pi(A_j)^{1/2}\pi(A_1)^{1/2}-(5/4)\lambda^2. \label{eqn:lowerEA_j'}
\end{eqnarray}

\begin{eqnarray}
\lefteqn{\E(\pi(A_j')|\mathcal C) \; = \; \pi(A_j) + R(V,A_j)-S_{\mathcal C}(A_j)} \nonumber\\
&\leq& \pi(A_j) +\pi(A_j)^2 + \lambda^2\pi(A_j)\pi(A_j^c)
- \pi(A_j)\sum_{i=1}^k \pi(A_i)^2 + 2\lambda\pi(A_j)^{1/2}\sum_{i=1}^k\pi(A_i)^{3/2} \nonumber\\
&\leq&\pi(A_j)\brac{1 + \pi(A_j)-\sum_{i=1}^k\pi(A_i)^2}
            +(1/4)\lambda^2 +2\lambda\pi(A_j)^{1/2}\pi(A_1)^{1/2}. \label{nk7r}
\end{eqnarray}
%
By assumption, $\lambda \le \pi(A_1)/32$ and $\pi(A_1) \le 2/3$,
so~\eqref{eqn:lowerEA_j'} and~\eqref{nk7r} imply
\begin{equation}\label{eqn:kdoow}
\pi(A_1)/2 \; \le \; \E(\pi(A_1')|\mathcal C) \; \leq \; 2 \pi(A_1).
\end{equation}
Define $\varepsilon_1 = \sqrt{\frac{9\log n}{\E(A_1'|\mathcal C)}}
\leq \sqrt{\frac{18\log n}{A_1}} <1$.
Therefore, using the Chernoff bounds~\eqref{eqn:k33dz}, we get
\begin{equation}\label{eqn::concentrationA_1'}
   \Prob(A_1'\leq (1-\varepsilon_1)\E(A_1'|\mathcal C) |\mathcal C) \leq e^{-3\log(n)} = n^{-3}.
\end{equation}
For a fixed $j$, $2 \le j \le k$,
define $\varepsilon_j = {\sqrt{9(\log n)\E(A_1'|\mathcal C)}}/{\E(A_j'|\mathcal C)}$. We show
that
\begin{eqnarray}\label{eqn:concentrationA_2'}
\Prob(A_j'\geq (1+\varepsilon_j)\E(A_j'|\mathcal C) |\mathcal C)
     \leq n^{-3}.
\end{eqnarray}
Indeed,
if $\varepsilon_j \le 1$, then the Chernoff bounds~\eqref{eqn:k33dz} give
\[
\Prob(A_j'\geq (1+\varepsilon_j)\E(A_j'|\mathcal C) |\mathcal C)
     \leq e^{-3(\log n)\E(A_1'|\mathcal C) / \E(A_j'|\mathcal C)}
    \leq e^{-3\log(n)}
    = n^{-3}.
\]
If $\varepsilon_j > 1$, then the Chernoff bound~\eqref{eqn:k99dz} gives
\[
\Prob(A_j'\geq (1+\varepsilon_j)\E(A_j'|\mathcal C) |\mathcal C)
     \leq e^{-\sqrt{(\log n)\E(A_1'|\mathcal C)}}
     \leq e^{-\sqrt{A_1}}
     \leq e^{-(Cn)^{1/3}}
    \leq  n^{-3}.
\]
%

The bounds~\eqref{eqn::concentrationA_1'} and~\eqref{eqn:concentrationA_2'} imply
that with probability at least $1  - kn^{-3} \ge 1 - n^{-2}$, for all $2 \le j \le k$,
\begin{eqnarray}
A_1'-A_j' &\geq& (1-\varepsilon_1)\E(A_1'|\mathcal C)-(1+\varepsilon_j)\E(A_j'|\mathcal C) \nonumber\\
&=&\E(A_1'-A_j'|\mathcal C) -2\sqrt{9(\log n) \E(A_1'|\mathcal C)}
\end{eqnarray}
and thus
\begin{equation}\label{jcc0gt6}
   \pi(A_1')-\pi(A_j') \; \geq \; \E(\pi(A_1')-\pi(A_j')|\mathcal C)
 -2\sqrt{\frac{9(\log n)\E(\pi(A_1')|\mathcal C)}{n}}.
\end{equation}
The right-hand side of~\eqref{nk7r} is non-increasing with increasing $j$, so for each
$2 \le j \le k$,
\begin{eqnarray}
\E(\pi(A_j')|\mathcal C)
   &\leq& \pi(A_2)(1 + \pi(A_2)-\sum_{i=1}^k\pi(A_i)^2) +\lambda^2
     +2\lambda\pi(A_1)\label{eqn:upperpi(A_j')}.
\end{eqnarray}
Let $\Delta = \pi(A_1)-\pi(A_2)$.
Inequalities~\eqref{eqn:lowerEA_j'} and~\eqref{eqn:upperpi(A_j')} give for each $2 \le j \le k$,
\begin{eqnarray}
\E(\pi(A_1')-\pi(A_j')|\mathcal C)
     &\geq& \pi(A_1)\left(1+\pi(A_1)-\sum_{i=1}^k\pi(A_i)^2 \right)
               - 2\lambda\pi(A_1)-(5/4)\lambda^2 \nonumber\\
&& - \left( \pi(A_2)\brac{1 + \pi(A_2)-\sum_{i=1}^k\pi(A_i)^2} +(1/4)\lambda^2
     +2\lambda\pi(A_1)\right) \nonumber\\
&=&  \Delta\brac{1+\pi(A_1)+\pi(A_2)-\sum_{i=1}^k \pi(A_i)^2}
        - 4\lambda\pi(A_1) -(3/2)\lambda^2\nonumber\\
&\geq& \Delta(1+ \pi(A_1)\pi(A_1^c))- 4\lambda \pi(A_1)-2\lambda^2 \label{klm3s1} \\
&\geq& \Delta+\Delta\pi(A_1)/7. \label{eqn:lowerEA_1-A_2}
\end{eqnarray}
Inequality~\eqref{klm3s1} holds because $\sum_{i=2}^k \pi(A_i)^2 \le \pi(A_2)$.
In the last step we used that $\pi(A_1^c)\ge 1/3$ and $\lambda \leq \Delta/32$.
From~\eqref{jcc0gt6}, \eqref{eqn:lowerEA_1-A_2} and~\eqref{eqn:kdoow},
with probability at least $1-n^{-2}$,
\begin{eqnarray}\nonumber
\min_{2 \le j \le k} \{\pi(A_1')-\pi(A_j')\} &\ge&
\E(\pi(A_1')-\pi(A_j')|\mathcal C)
-\frac{\ve_1}{n}\E(A_1'|\mathcal C)
-\frac{\ve_j}{n}\E(A_j'|\mathcal C)\\
 &\geq& \Delta(1+\pi(A_1)/7)-6\sqrt{\frac{2\log n}{n}\pi(A_1)}\nonumber\\
&=& \Delta\left(1+\pi(A_1)/7-\frac{6}{\Delta}\sqrt{\frac{2\log n}{n}\pi(A_1)}.\right)\nonumber.
\end{eqnarray}
By assumption,
$\Delta \ge 240\sqrt{2\log(n)/A_1}$, so with probability at least $1-n^{-2}$,
\begin{eqnarray}
\min_{2 \le j \le k} \{\pi(A_1')-\pi(A_j')\} &\geq& \Delta(1+\pi(A_1)/10),\label{eqn:Delta-2}
\end{eqnarray}
and we get we get~\eqref{eqn:Delta}.
This also proves
that \whp\ opinion 1 remains the majority opinion.
The order between the other opinions might change.

To get information about the increase in the number of vertices with opinion 1,
we use  Equation~\eqref{eqn:lowerEA_j'} with $j=1$ and the assumption that $\lambda \leq \Delta/32$. We  obtain
\begin{eqnarray}
\E(\pi(A_1')|\mathcal C) &\geq& \pi(A_1)(1+\pi(A_1)-\sum_{i=1}^k \pi(A_i)^2)-\Delta\pi(A_1)/16-\Delta^2/(32)^2 \nonumber\\
&\geq& \pi(A_1)(1+\pi(A_1)-\pi(A_1)^2-\pi(A_2)\pi(A_1^c) - \Delta/16-\Delta/(32)^2)\nonumber\\
&>& \pi(A_1)(1+\Delta/4). \label{eqn:yh76}
\end{eqnarray}
By using Chernoff bounds~\eqref{eqn:k33dz} with
$\varepsilon = \sqrt{\frac{9\log n}{\E(A_1'|\mathcal C)}}$
and Inequalities~\eqref{eqn:yh76} and~\eqref{eqn:kdoow},
with probability at least $1-n^{-2}$,
\begin{eqnarray}
A_1' &\geq& A_1(1+\Delta/4)-\sqrt{\E(A_1'|\mathcal C)9\log n} \; \ge \; A_1(1+\Delta/4)-\sqrt{18A_1\log n}\nonumber\\
&=& A_1(1+\Delta/4-3\sqrt{2}\sqrt{\log n/A_1)}).\label{jkdbua}
\end{eqnarray}
From the assumptions of the lemma, we have
$\Delta/20 = (A_1-A_2)/(20n) \ge 3\sqrt{2}\sqrt{\log n/A_1}$.
Therefore~\eqref{jkdbua} implies
$A_1' \ge A_1(1+\Delta/5)$, which is the same as~\eqref{xi}.
\end{proof}

\begin{sloppy}

\begin{lemma}\label{lem:main}
Assume $A_1 \le 2n/3$, $A_1-A_2 \geq C n \sqrt{(\log n )/A_1}$
(requiring $A_1 \ge C^{2/3} n^{2/3} \log^{1/3} n$),
where $C = 240 \sqrt{2}$, and $\lambda \leq {(A_1-A_2)}/{(32n)}$.
Then with probability at least
$1- 1/n$,
after at most
$O((n/A_1) \log(A_1/(A_1 - A_2)))$
rounds, the number of vertices with opinion 1 is at least $2n/3$.
\end{lemma}

\end{sloppy}

\begin{proof}
We apply Lemma~\ref{lem:main-oneStep} to consecutive rounds until the size of opinion 1 reaches $2n/3$.
Since \whp\ the difference between the size of opinion 1 and the size of the second largest opinion increases,
our assumption about $\lambda$ in Lemma~\ref{lem:main-oneStep} is maintained from round to round.
At the end of each round the ordering of the opinions according to their sizes can change.
In that case we exchange the labels of the opinions so that $A_1(t) \geq A_2(t) \cdots \ge A_k(t) $ for every round $t$.
Lemma~\ref{lem:main-oneStep}, however, implies
that \whp\ opinion $1$ remains the largest opinion, so it is not relabeled.

Denote by $x(i)$ the fraction of vertices with opinion 1 at the end of round $i$, where $x(0) = \pi(A_1)$,
and by $y(i)$ the difference between the fraction of vertices with opinion 1 and the fraction of vertices with
the second largest opinion, where $y(0) = \Delta = \pi(A_1) - \pi(A_2) < x(0)$.
By~\eqref{xi} and~\eqref{eqn:Delta} and induction on the number of rounds,
with probability at least $1-1/n$,
for each round $1 \le i \le n$, if $x(i) < 2/3$, then
\begin{eqnarray}
x(i) & \geq & x(i-1)(1+y(i-1)/5), \label{hjjowb-x}\\
y(i) & \geq & y(i-1)(1+x(i-1)/10).\label{hjjowb-y}
\end{eqnarray}

Iterating~\eqref{hjjowb-x} and~\eqref{hjjowb-y} for $j = \lceil 10/x(0) \rceil < n$ rounds, we get
$y(j) \ge 2y(0)$ and $x(j) \ge x(0) + y(0)$, or
$x(i) \ge 2/3$ for some $i \le j$.
Repeating this $r= \lceil \log_2(x(0)/y(0)) \rceil$ times, we get for round $i_1 = rj < n$,
$y(i_1) \ge x(0)$ and
$x(i_1) \ge x(0) + y(0) + 2y(0) + 4y(0) \cdots + 2^{r-1}y(0) \ge 2 x(0)$,
or $x(i) \ge 2/3$ for some $i \le i_1$.

If for some $q \ge 1$, $y(i_q) \ge 2^{q-1} x(0)$ and
$x(i_q) \ge 2^q x(0)$, or $x(i) \ge 2/3$ for some $i \le i_q$,
then at the end of round $i_{q+1} = i_q + \lceil 10/(2^q x(0)) \rceil$,
$y(i_{q+1}) \ge 2^{q} x(0)$ and
$x(i_{q+1}) \ge 2^{q+1} x(0)$, or
$x(i) \ge 2/3$ for some $i \le i_{q+1}$, or $i_{q+1} > n$.
Taking $q = \lceil \log_2 (1/x(0)) \rceil$, we have
$i_q = O((1/x(0)) \log(x(0)/y(0))) = O((n/A_1) \log(A_1/(A_1 - A_2)))$
(observe that $i_q < n$) and $2^q x(0) \ge 1$, so we must have $x(i) \ge 2/3$ for some $i \le x(i_q)$.
\end{proof}

When the largest opinion reaches the size $2n/3$, it will take over the whole graph
within additional $O(\log n)$ rounds.
The progress of  voting in this final stage would be slowest, if all minority opinions were joined
together into a single ``second'' opinion.
The proof of the next lemma follows the proof from~\cite{cooper2015fast}
that two-sample voting finishes in $O(\log n)$ rounds, if there are two opinions,
the majority opinion has size at least $cn$, for a constant $c > 1/2$, and $\lambda$ is sufficiently small.


\begin{lemma}\label{lemma:main2}
Let $G$ be a connected regular graph with $\lambda \leq 1/4$.
If the majority opinion has size at least $2n/3$, then
with probability at least $1 - n^{-2}$,
the voting finishes within $\mathcal O(\log n)$ rounds.
\end{lemma}
\begin{proof}
Let $A$ represent the current set of vertices with the majority opinion.
We put all minority opinions into one opinion set
$B = V \setminus A$ and analyse two-sample voting with these two opinions.
The majority opinion in this process is always a subset of the majority opinion in the original process,
when there are distinct minority opinions.

Let $A'$ and $B'$ be the corresponding sets in the next round.
We compute $\E(A'|A)$. Observe that since
in our context $\mathcal C = (A,B)$ and
$S_{\mathcal C}(A) = R(A,A)+R(B,A)$, then, from Equation~\eqref{eqn:ExpectedDiff}, we have
\begin{eqnarray}
\E(\pi(B')|B) &=&\pi(B)+R(V,B)-S_{\mathcal C}(B) \nonumber\\
&\leq& \pi(B) + \pi(B)^2+\lambda^2 \pi(B)\pi(A) - \sum_{x \in B} \pi(x)(P(x,A)^2+P(x,B)^2) \nonumber\\
&\leq& \pi(B) + \pi(B)^2+\lambda^2 \pi(B)\pi(A) - \pi(B)/2 \nonumber\\
&=& \pi(B) + \pi(B)(1/2-(1-\lambda^2)\pi(A)). \label{bn9os6q}
\end{eqnarray}
Given $\lambda \leq 1/4$ and $\pi(A) \geq 2/3$, \eqref{bn9os6q} implies
\begin{equation}\label{fghw0pqy6}
   \E(\pi(B')|B) \; \leq \; (7/8)\pi(B).
\end{equation}
A standard coupling shows that if $B_1 \subseteq  B_2$,
then $\Prob(\pi(B')\geq \delta \, |\,  B = B_1) \leq \Prob(\pi(B') \geq \delta\,|\, B= B_2)$.
Take arbitrary sets $B_1\subseteq B_2 \subseteq V$ such that $\pi(B_1) \leq 1/3$
and $\pi(B_2) = 1/3$, and apply
Hoeffding's Inequality
to get
\begin{eqnarray}
{\Prob((\pi(B') \geq 1/3)\:|\:B = B_1)}
& \le & \Prob(\pi(B') \geq 1/3\:|\:B = B_2)\nonumber\\
& = &  \Prob(|B'| \geq n/3 \:|\: B = B_2)\nonumber\\
& \le & \Prob(|B'| \geq \E(|B'|\, | B=B_2) + n/24 \:|\: B = B_2) \label{lkr8w} \\
& \le & e^{-2 (n/24)^2/n} = o(n^{-10}). \label{ml0w2s}
\end{eqnarray}
Inequality~\eqref{lkr8w} holds because $\E(|B'|\, | B=B_2) \le (7/8)\pi(B_2) = (7/24)n$,
and~\eqref{ml0w2s} follows from Hoeffding's Inequality.
The  bound above implies
that in the next $n$ rounds, the probability
to have a minority with more than $n/3$ opinions is $o(n^{-9})$.

Let $B_t$ be the set with the minority opinion after
$t$ rounds of this final stage of voting.
We assume that $B_0$ is a fixed set such that $|B_0| \le (1/3)n$.
To obtain the claimed bound on the number of rounds,
we use~\eqref{fghw0pqy6} and~\eqref{ml0w2s} in a straightforward application of Markov's Inequality:
\begin{equation}\label{jklr56cv}
 \Prob(B_t \neq \emptyset) \; = \; \Prob(\pi(B_t) \geq 1/n) \; \le \; n \cdot \E(\pi(B_t)).
\end{equation}
Using~\eqref{fghw0pqy6}, for each $t \ge 1$,
\[
   \E(\pi(B_t)|B_{t-1}) \; \leq \; \left\{ \begin{array}{ll}
       (7/8)\pi(B_{t-1}), & \mbox{if $B_{t-1} \le 1/3$,} \\
        1, & \mbox{if $B_{t-1} > 1/3$.}
      \end{array} \right.
\]
This gives
\[  \E(\pi(B_t)) \; = \; \E(\E(\pi(B_t)|B_{t-1}))
     \; \le \; (7/8) \E(\pi(B_t)) + \Prob(B_{t-1} > 1/3),
\]
Applied the above inequality iteratively to obtain
\begin{eqnarray}
 \E(\pi(B_t))
     & \le & (7/8)^t \pi(B_0) + \sum_{\tau = 0}^{t-1} \Prob(B_{\tau} > 1/3)
     \; \le \; (1/3) \cdot (7/8)^t + o(n^{-8}). \nonumber
\end{eqnarray}
Thus for $T = K\log n$ with $K =4/\log(8/7)$,
$\E(\pi(B_T)) \le n^{-3}$, so~\eqref{jklr56cv}
implies that with probability at least $1-n^{-2}$,
$B_T$ is empty, that is, the voting finishes in $K\log(n)$ rounds.
\end{proof}

\section{Reducing Three-sample voting to Two-sample voting}
\label{3to2}

In this section we study the three-sample voting process, which is similar to the two-sample voting process but  samples three neighbours in each round. Additionally, if all three opinions are distinct, the vertex adopts the opinion of the first sampled neighbour.
Formally, each vertex $v$ selects three random neighbours with replacement and considers their opinions, say, $Y_{v,1}, Y_{v,2}, Y_{v,3}$.
Vertex $v$ changes its opinion to the majority of $\{Y_{v,1}, Y_{v,2}, Y_{v,3}\}$, or, if there is no majority, to $Y_{v,1}$.
Suppose in a given round we have $k$ opinions. Let $\mathcal C = (A_1, \ldots, A_k)$ be the partition of the vertices given by the opinions, where $A_j$ is the set of vertices with opinion $j$. Let $A_j''$ be the vertices with opinion $j$ at the next round.
Moreover, let $A_j'$ be the set of vertices $v$ such that $Y_{v,1}=j$.

The following lemma will allows us to use the results of Lemma~\ref{lem:main} and Lemma~\ref{lemma:main2} for the three-sample protocol.
Due to space restrictions the proof of the lemma, and the explanation of its application in  Lemma~\ref{lem:main} and Lemma~\ref{lemma:main2} is given in the Appendix.

\begin{lemma}\label{lemma:transfermain}
Let $G$ be a connected graph and let $\mathcal C = (A_1,\ldots, A_k)$ partition $V$. Then
\begin{eqnarray}\label{eqn:Expected''}
\E(\pi(A_j'')| \mathcal C) = \pi(A_j)+R(V,A_j)-\E(S_{\mathcal C}(A_j')| \mathcal C)
\end{eqnarray}
Moreover,
\begin{eqnarray}\label{eqn:lowerE(S_C(A_j'))}
\E(S_{\mathcal C}(A_j')|\mathcal C) \geq \pi(A_j)\sum_{i=1}^k \pi(A_i)^2-2\lambda \pi(A_j)^{1/2}\sum_{i=1}^k \pi(A_i)^{3/2}
\end{eqnarray}
and
\begin{eqnarray}\label{eqn:upperE(S_C(A_j'c))}
\E(S_{\mathcal C}(A_j'^c)|\mathcal C) \geq \pi(A_j^c)\sum_{i=1}^k \pi(A_i)^2-2\lambda \pi(A_j^c)^{1/2}\sum_{i=1}^k \pi(A_i)^{3/2}.
\end{eqnarray}
If $\mathcal C = (A,B)$, then
 \begin{eqnarray}\label{eqn:upperE(S_C(B'))}
 \E(S_{\mathcal C}(B')|\mathcal C) \geq \pi(B)/4.
 \end{eqnarray}
\end{lemma}

Before proving Lemma \ref{lemma:transfermain}, we
observe that Equations~\eqref{eqn:lowerE(S_C(A_j'))} and~\eqref{eqn:upperE(S_C(A_j'c))} are enough to get
Lemma~\ref{lemma:BoundsS_c(A)} for the values $\E(S_{\mathcal C}(A_j')|\mathcal C)$,
i.e.\ the bounds we got for $S_{\mathcal C}(A_j)$ are also valid for
$\E(S_{\mathcal C}(A_j')|\mathcal C)$.
Our proof of Lemma~\ref{lem:main} is based on the concentration of sums of
Bernoulli random variables around their expected values, but we see that the expected values,
or, more precisely, the respective bounds on those values, are the same in both protocols.
Thus the ``w.h.p.'' result of Lemma~\ref{lem:main} applies also to the three-sample voting model.
The same argument but using Equation~\eqref{eqn:upperE(S_C(B'))} allows us to transfer the result of Lemma~\ref{lemma:main2}
from the two-sample to the three-sample voting.

\begin{proof}
First of all, observe that $A_j'$ is the result of choosing only one vertex, i.e. one round of standard pull voting.  For given vertex $v$ this accounts for $Y_{v,1}$. We now consider $Y_{v,2}, Y_{v,3}$ taken in the original partition $\cal C$.
Observe that given $A_j'$, then $A_j''$ is the set of vertices in $A_j'$ such that the other two opinions
taken in the original partition $\cal C$ are not equal to any opinion $i$ other than $j$, plus the set of vertices outside $A_j'$ such that the other two opinions in $\cal C$ are equal to $j$. Therefore
\begin{eqnarray}
\pi(A_j'') = \pi(A_j')+\pi(\{x \in A_j'^c: Y_{x,2} = Y_{x,3} = j\}) - \pi(\{x \in A_j': Y_{x,2} = Y_{x,3} = i, i \neq j\})
\end{eqnarray}
By a result of \cite{HassinPeleg-InfComp2001} for classical pull voting, we have $\pi(A_j'|\mathcal C) = \pi(A_j)$.
From there, it is relatively straightforward to get that
\begin{eqnarray}
\E(\pi(A_j'')|\mathcal C) &=& \pi(A_j)+\E\left(\sum_{x \in A_j'^c} \pi(x)P(x,A_j)^2- \sum_{x \in A_{j}'} \pi(x)\sum_{i \neq j}P(x,A_i)^2 \;\middle| \;\mathcal C\right)\nonumber\\
&=&\pi(A_j)+\E\left(\sum_{x \in V} \pi(x)P(x,A_j)^2- \sum_{x \in A_{j}'} \pi(x)\sum_{i=1}^kP(x,A_i)^2 \;\middle| \;\mathcal C\right)\nonumber\\
&=& \pi(A_j) + R(V,A_j)-\E(S_{\mathcal C}(A_j')|\mathcal C)
\end{eqnarray}
For the lower bound in \eqref{eqn:lowerE(S_C(A_j'))} we use Lemma~\ref{lemma:BoundsS_c(A)} to get
\begin{eqnarray}
S_{\mathcal C}(A_j') \geq \pi(A_j')\sum_{i=1}^ k \pi(A_i)^2 -2\lambda \pi(A_j')^ {1/2}\sum_{i=1}^k \pi(A_i)^{3/2}.
\end{eqnarray}
By concavity of $f(x)=x^{1/2}$ we have
\[
\E(\pi(A_j')^{1/2}|\mathcal C) \leq (\E( \pi(A_j') | \mathcal C )^{1/2}= (\pi(A_j))^{1/2},
\]
obtaining the result of Equation~\eqref{eqn:lowerE(S_C(A_j'))}. A similar argument gives us the result of Equation~\eqref{eqn:upperE(S_C(A_j'c))}.
\end{proof}

\ignore{

\section{$l$-neighbourhoods}

Suppose we have a connected regular graph $G$ with expansion $\lambda$ and suppose that $\lambda$
does not satisfy the conditions to apply our main theorem.
Then we can consider that instead of selecting two (or three) neighbours, each vertex starts two (or three)
independent $\ell$-step random walks, stopping in vertices $v_1$ and $v_2$
(or $v_1$, $v_2$ and $v_3$) and performing the two-sample (or three-sample) decision.
Since we are using $\ell$ steps of a random walk,
this is equivalent to running the protocol with transition matrix $P^\ell$.
Since the stationary distribution is the same for $P$ and $P^l$, the result of
Theorem~\ref{thm:mainResult} translates immediately to this case. The only difference is that $\lambda(P^l) = \lambda(P)^l$,
thus weakening the conditions on $\lambda$.

}



\end{document}